\documentclass[11pt,fullpage]{article}

\usepackage{url}
\usepackage{fullpage}
\usepackage{xspace}
\usepackage{amsthm}
\usepackage{cite}
\usepackage{amsfonts}
\usepackage{relate}
\usepackage{times}
\usepackage{epstopdf}
\usepackage{defs}
\usepackage{clrscode}
\usepackage{amsmath}
\usepackage{color}
\usepackage[margin=1in]{geometry}

\newcommand{\Ins}{\mathcal{I}}
\newcommand{\Top}{\mathcal{T}}
\newcommand{\eps}{\varepsilon}
\newcommand{\E}{\mathbb{E}}
\newcommand{\M}{\textbf{M}}
\newcommand{\rank}{\textrm{rank}}
\renewcommand{\log}{\lg}

\begin{document}

\begin{titlepage}

\title{Approximate Range Emptiness in Constant Time and Optimal Space}

\author{
  Mayank Goswami\thanks{Max-Planck Institute for Informatics.
Email: \texttt{gmayank@mpi-inf.mpg.de}.}
\and
 Allan Gr\o nlund\thanks{Aarhus University.
Email: \texttt{jallan@cs.au.dk}. Supported by MADALGO - Center for Massive Data
  Algorithmics, a Center of the Danish National Research Foundation.}
\and
  Kasper Green Larsen\thanks{Aarhus University.
Email: \texttt{larsen@cs.au.dk}. Supported by MADALGO - Center for Massive Data
  Algorithmics, a Center of the Danish National Research Foundation.}
  \and 
  Rasmus Pagh\thanks{IT University of Copenhagen.
Email: \texttt{pagh@itu.dk}. Supported by the Danish National Research
  Foundation under the Sapere Aude program.}
}

\date{}

\maketitle

\begin{abstract}
This paper studies the \emph{$\varepsilon$-approximate range emptiness} problem, where the task is to represent a set $S$ of $n$ points from $\{0,\ldots,U-1\}$ and answer emptiness queries of the form ``$[a ; b]\cap S \neq \emptyset$ ?'' with a probability of \emph{false positives} allowed.
This generalizes the functionality of \emph{Bloom filters} from single point queries to any interval length $L$.
Setting the false positive rate to $\varepsilon/L$ and performing $L$ queries, Bloom filters yield a solution to this problem with space $O(n \lg(L/\varepsilon))$ bits, false positive probability bounded by $\varepsilon$ for intervals of length up to $L$, using query time $O(L \lg(L/\varepsilon))$.
Our first contribution is to show that the space/error trade-off cannot be improved asymptotically: 
Any data structure for answering approximate range emptiness queries on intervals of length up to $L$ with false positive probability $\varepsilon$, must use space $\Omega(n \lg(L/\varepsilon)) - O(n)$ bits.
On the positive side we show that the query time can be improved greatly, to constant time, while matching our space lower bound up to a lower order additive term.
This result is achieved through a succinct data structure for (non-approximate 1d) range emptiness/reporting queries, which may be of independent interest.
\end{abstract}

\thispagestyle{empty}
\end{titlepage}
\newpage

\section{Introduction}
The approximate membership problem is a fundamental and well-studied
data structure problem. Here we are to represent a static set
$S$ of $n$ distinct points/elements from a universe
$[U]=\{0,\dots,U-1\}$. A query is specified by a point $x \in [U]$
and the data structure must answer \emph{member} with probability
$1$ if $x \in S$. If $x \notin S$, then the data structure must answer
\emph{not member} with probability at least $1-\eps$, where 
$\eps > 0$ is a parameter of the data structure and the
probability is over the randomness used when constructing the data structure. 
Hence the name
approximate membership. Note that solutions are inherently 
randomized as queries for points not in $S$ must be answered
correctly with a given probability. The value $\eps$ is typically
refered to as the \emph{false positive rate}.

While storing $S$ directly requires $\lg \binom{U}{n} \ge n\lg(U/n)$
bits, approximate membership data structures can be implemented to use
only $n\lg(1/\eps) + O(n)$ bits
, which is known to be
optimal~\cite{Carter}. Thus efficient approximate membership data structures
use less space than the information theoretic minimum for storing $S$
directly.

In this paper, we study a natural generalization of the approximate
membership problem, in which we are to handle query intervals rather than single-point
queries. More formally, we define the \emph{approximate range
  emptiness} problem as follows: Represent a set $S$ of $n$ points
from a universe $[U]$, such that given a query interval $I = [a ; b]$,
the data structure must answer \emph{non-empty} with probability $1$
if there is a point from $S$ in $I$ (i.e. ``is $S \cap I \neq \emptyset$?''). If
$I$ contains no points from $S$, the data structure must answer
\emph{empty} with probability at least $1-\eps$.

The approximate range emptiness problem was first considered in database applications~\cite{alexiou2013adaptive}. Here approximate range emptiness data structures were used to store a small (approximate) representation of a collection of points/records in main memory, while maintaining the actual (exact) points on disk. When processing range queries, the approximate representation is first queried to avoid expensive disk accesses in the case of an empty output. The paper takes a heuristic approach to the problem and design data structures that seem to perform well in practice. Unfortunately no formal analysis or worst case performance guarantees are provided. Motivated by this lack of theoretical understanding, we ask the following question:
\begin{question}
What theoretical guarantees can be provided for the approximate range emptiness problem?
\end{question}
Towards answering this question, observe that an approximate range emptiness structure solves the approximate membership problem and hence a space lower bound of $n\lg(1/\eps)$ bits follows directly from previous work. But can this space bound be realized? Or does an approximate range emptiness data structure require even more space? What if we require only that the data structure answers queries of length no more than a given input parameter $L$?  At least for $L \ll u/n$, we can beat the trivial $n \lg(U/n)$ space bound as follows: Implement the approximate membership data structure of~\cite{pagh2005optimal} on $S$ with false positive rate $\eps/L$. Upon receiving a query interval $I$ of length no more than $L$, query the approximate membership data structure for every $x \in I \cap [U]$ and return \emph{non-empty} if any of these queries return \emph{member}. Otherwise return \emph{empty}. By a union bound, this gives the desired false positive rate of $\eps$ and the space consumption is $n\lg(L/\eps)$ bits. How much can we improve over this trivial solution? A natural line of attack would be to design an approximate membership structure where the locations of the false positives are correlated and tend to ``cluster'' together inside a few short intervals. Is this possible without increasing the space usage? Answering these questions is the focus of this paper.

\paragraph{Our Results.}
In Section~\ref{sec:lower}, we answer the questions above in the strongest
possible negative sense, i.e. we prove that any data structure for the
approximate range emptiness problem, where only query intervals of length up
to a given $L$ are to answered, must use
$$
n\lg\left(\frac{L^{1-O(\eps)}}{\eps}\right) - O(n)
$$
bits of space. Thus it is not even possible to shave off a constant factor in
the space consumption of the trivial solution, at least for
$\eps=o(1)$. Moreover, the lower
bound applies even if only queries of length exactly $L$ are to be
answered. We find this extremely surprising (and
disappointing).

In light of the strong lower bounds, we set out to improve over the $\Omega(L)$ query time of the trivial solution above, while maintaining optimal space. In Section~\ref{sec:upper} we present a data structure with $O(1)$ query time and space $n \lg(L/\eps) +o(n \lg(L/\eps)) + O(n)$ bits, thus matching the lower bound up to a lower order term whenever $L/\eps = \omega(1)$. 
The data structure answers a query of any length $\ell \leq L$, not only those of length exactly $L$, with false positive rate $\varepsilon \ell/L$. 
As a building block in our data structure, we also design a new succinct data structure for (non-approximate) range emptiness queries which may be of independent interest. This data structure uses $n\lg(U/n)+O(n\lg^{\delta}(U/n))$ bits of space for storing $n$ points from a universe $[U]$, while answering queries in constant time. Here $\delta>0$ is an arbitrarily small constant. The data structure is thus optimal up to the lower order additive term in the space usage. Moreover, it also supports reporting all $k$ points inside a query interval in $O(k)$ time, thus providing a succinct 1d range reporting data structure. The best previous data structure was a non-succinct $O(n\lg U)$ bit data structure of Alstrup et al.~\cite{alstrup1d}.

As an additional result, we also prove in Section~\ref{sec:twosided} that
data structures with two-sided error $\eps>0$ (i.e., for non-empty
intervals, we must answer \emph{non-empty} with probability $1-\eps$), must
use space
$$
n \lg(L/\eps) - O(n)
$$
bits when $0 < \eps < 1/\lg U$. Thus for small error rate, only
lower order additive savings in space are possible. For $1/\lg U < \eps <
1/2-\Omega(1)$, we get a space lower bound of
$$
\Omega\left(\frac{n \lg(L \lg U)}{\lg_{1/\eps} \lg U}\right)
$$
bits, thus ruling out hope of large space savings even with two-sided
errors. Again, these lower bounds hold even if only query intervals of
length exactly $L$ are to be answered.

\paragraph{Related Work on Approximate Membership.} 
Bloom filters~\cite{bloom1970space} are the first approximate membership data structures, requiring $n\lg(1/\eps)\lg e$ space and having a lookup time of $\lg(1/\eps)$. They have found a variety of applications both in theory and practice, and we refer the reader to~\cite{broder2004network} for an overview of applications. In \cite{pagh2005optimal} the space usage was reduced to near-optimal $(1+o(1)) n \lg (1/\eps)$ bits and the lookup time to an optimal $O(1)$, while also supporting insertions/deletions in amortized expected $O(1)$ time. In the static case the space usage has been further reduced to $o(n)$ bits from the lower bound~\cite{belazzougui2013compressed}.

The dynamic case where the size $n$ of the set $S$ is not known in advance was handled in~\cite{pagh2013approximate}, where it was shown that the average number of bits per element in an optimal data structure must grow with $\lg \lg n$. A closely related problem to approximate membership is the retrieval problem, where each element has an associated data, and data must be retrieved correctly only for members. In~\cite{belazzougui2013compressed} the authors achieve a query time of $O(1)$ using a space that is within $o(n)$ of optimal. 
%

\section{Lower Bounds}
\label{sec:lower}
In this section, we prove a lower bound on the space needed for any
data structure answering approximate range emptiness queries. While
the upper bounds presented in Section~\ref{sec:upper} guarantees a
false positive rate of $\eps$ for any query interval of length up to a
predefined value $L$, the lower bound applies even if \emph{only}
length $L$ intervals are to be answered. More formally, we show:

\begin{theorem}
\label{thm:lower}
For the approximate range emptiness problem on $n$ points from a
universe $[U]$, any data structure answering all query intervals of a
fixed length $L \leq u/5n$ with false positive rate $\eps>0$, must use
at least
$$s \geq n \lg \left( \frac{L^{1-O(\eps)}}{\eps} \right)-O(n)$$
bits of space.
\end{theorem}

The proof of Theorem~\ref{thm:lower} is based on an encoding
argument. The high level idea is to use an approximate range emptiness
data structure to uniquely encode (and decode) every set of $n$ points
into a bit string whose length depends on the space usage and false
positive rate of the data structure. Since each point set is uniquely
encoded, this gives a lower bound on the size of the encoding and
hence the space usage of the data structure. 

For technical reasons, we do not encode every set of $n$ points, but
instead only point sets that are well-separated in the following
sense: Let $\Ins$ be the set of all $L$-\emph{well-separated} sets of
$n$ points in $[U]$, where a set $S$ of $n$ points is
$L$-well-separated if:
\begin{itemize}
\item For any two distinct $x,y \in S$, we have $|x-y| \geq 2L$.
\item For any $x \in S$, we have $x \geq 2L-1$ and $x \leq U-2L$.
\end{itemize}
Note that there are many $L$-well-separated when $L$ is not too close
to $U/n$:
\begin{lemma}
\label{lem:manywell}
There are at least $((U-4nL)/n)^n$ $L$-well-separated point sets
for any $L \leq U/4n$.
\end{lemma}
\begin{proof}
Consider picking one point at a time from $U$, each time making sure
the constructed set is $L$-well-separated. For the $i$'th point,
$i=1,\dots,n$, there are $U-4L-4(i-1)L = U-4iL$ valid choices for
the location of the $i$'th point. Since the same set of points can be
constructed in $n!$ ways using this procedure, there are at least 
$$\frac{(U-4nL)^n}{n!} \geq \left(\frac{u-4nL}{n}\right)^n.$$
$L$-well-separated point sets.
\end{proof}

We now set out to encode any set $S \in \Ins$ in a number of bits
depending on the performance of a given data structure. Assume for
simplicity that $U$ and $L$ are powers of $2$. The encoding algorithm
considers two carefully chosen sets of intervals:
\begin{itemize}
\item The \emph{top} intervals is the set of $U/L$ length $L$
  intervals $\Top = \{T_0,\dots,T_{U/L-1}\}$ of the form $T_i = [iL ;
    (i+1)L-1]$.
\item For each $x \in S$ and $i=1,\dots,\lg L-1$, the
  \emph{level-$i$-covering} intervals of $x$ are the two intervals:
\begin{eqnarray*}
  \ell_i(x) = [\lfloor x / 2^i \rfloor \cdot 2^i+2^{i-1}-L; \lfloor x / 2^i \rfloor \cdot 2^i+2^{i-1}-1]. \\
  r_i(x) = [\lfloor x/ 2^i \rfloor \cdot 2^i+2^{i-1} ; \lfloor x/2^i \rfloor \cdot 2^i +2^{i-1} + L -1].
\end{eqnarray*}
\end{itemize}
For intuition on why we take interest in these
intervals, observe the following:
\begin{lemma}
\label{lem:queryprops}
Consider the binary representation of a point $x \in S$. If the
$i$'th least significant bit of $x$ is $0$ (counting from $i=1$), then
$x$ is contained in $\ell_i(x)$. Otherwise $x$ is contained in
$r_i(x)$. Furthermore $\ell_i(x)$ and $r_i(x)$ contain no other points
from $S$ and any top interval $T_j$ contain at most one point from $S$.
\end{lemma}
\begin{proof}
Rewrite $x = \lfloor x/2^i \rfloor \cdot 2^i + (x \mod 2^i)$. If the
$i$'th least significant bit of $x$ is $0$ then $(x \mod 2^i) \in \{0,\dots,
  2^{i-1}-1\}$, i.e. $x \in [\lfloor x/2^i \rfloor \cdot 2^i ; \lfloor x/2^i \rfloor \cdot 2^i + 2^{i-1}-1] \subseteq \ell_i(x).$ Otherwise $(x \mod 2^i) \in \{2^{i-1}, \dots, 2^i-1\}$ implying $x \in [\lfloor x/2^i \rfloor \cdot 2^i + 2^{i-1} ; \lfloor x/2^i \rfloor \cdot 2^i + 2^i -1] \subseteq r_i(x).$

Next, observe that the intervals $\ell_i(x)$ and $r_i(x)$ have length
$L$ and both have an endpoint of distance less than $2^{i+1} \leq L$
from $x$. Since $S$ is $L$-well-separated, it follows that no other
points can be contained in $\ell_i(x)$ and $r_i(x)$. For the top
intervals $T_j$, the claimed property similarly follows from $S$ being
$L$-well-separated.
\end{proof}

We are ready to give the encoding and decoding algorithms for all
$L$-well-separated point sets $S \in \Ins$. For the encoding and
decoding procedures, we assume the existence of a data structure $D$
with $s$ bits of space and false positive rate $\eps>0$ for query
intervals of length $L$.

\paragraph{Encoding Algorithm.}
In this paragraph we present the encoding algorithm. Let $S \in \Ins$
be an $L$-well-separated point set. Observe that if we run the
randomized construction algorithm of $D$ on $S$, we are returned a
(random) memory representation $\M \in \{0,1\}^s$. For the memory
representation $\M$, the answer to each query is fixed, i.e. it is the
randomized choice of $\M$ that gives the false positive rate of $\eps$
for each query interval of length $L$ (the memory representation
encodes e.g. a particular hash function from a family of hash
functions). Letting $A(\M)$ denote the number of false positives
amongst the top intervals $\Top$, we get from linearity of expectation
that $\E[A(\M)] \leq \eps U/L$. Similarly, let $B(\M)$ denote the
number of false positives amongst all the intervals $\ell_i(x)$ and
$r_i(x)$ for $x \in S$ and $i \in \{1,\dots,\lg L-1\}$. We have
$\E[B(\M)] \leq \eps n \lg L$ since precisely $n \lg L$ of the
intervals $\ell_i(x)$ and $r_i(x)$ are empty. Since both $A(\M)$ and
$B(\M)$ are non-negative, it follows from Markov's inequality and a
union bound that:
$$\Pr_{\M}\left[A(\M) \leq \frac{1}{1-1/\gamma} \cdot \eps U/L
  \bigwedge B(\M) \leq \gamma \eps n \lg L\right] > 0$$ for any
parameter $\gamma>1$. Since the probability is non-zero, there exists
a particular memory representation $M^* \in \{0,1\}^s$ for which:
$$A(M^*) \leq \frac{1}{1-1/\gamma} \cdot \eps U/L \bigwedge B(M^*) \leq \gamma \eps n \lg L.$$
The encoding consists of the following:
\begin{enumerate}
\item The $s$ bits of $M^*$.
\item Let $\Top^*$ be the subset of $\Top$ that return
  \emph{non-empty} on $M^*$. We encode the set of $n$ top intervals
  $\Top_S$ containing $S$, where $\Top_S$ is specified as an $n$-sized
  subset of $\Top^*$. This costs $\lg \binom{|\Top^*|}{n} = \lg
  \binom{A(M^*)+n}{n}$ bits.
\item For each interval $T_j \in \Top_S$ in turn (from left to right),
  let $x$ be the point from $S$ in $T_j$. For $i=\lg L-1, \dots,1$ in
  turn, check whether both $\ell_i(x)$ and $r_i(x)$ return
  \emph{non-empty} on $M^*$. If so, we append one bit to our
  encoding. This bit is $0$ if $x \in \ell_i(x)$ and it is $1$ if $x
  \in r_i(x)$. Otherwise we continue without writing any bits to the
  encoding. In total this costs exactly $B(M^*)$ bits (each bit can be
  charged to exactly one false positive).
\end{enumerate}

This concludes the description of the encoding algorithm. Next we show
that $S$ can be uniquely recovered from the encoding.

\paragraph{Decoding Algorithm.}
In the following, we show how we recover an $L$-well-separated point
set from the encoding described in the paragraph above. The decoding
algorithm is as follows:
\begin{enumerate}
 \item Read the $s$ first bits of the encoding to recover $M^*$.
 \item Run the query algorithm for every query in $\Top$ with the
   memory representation $M^*$. This recovers $T^*$.
\item From $T^*$ and the bits written in step 2 of the encoding
  algorithm, we recover $T_S$.
\item For each interval $T_j \in T_S$ (from left to right), let $x$ be
  the point from $S$ in $T_j$. This step recovers $x$ as follows: From
  $T_j$, we know all but the $\lg L-1$ least significant bits of
  $x$. Now observe that the definition of $\ell_{\lg L -1}(x)$ and
  $r_{\lg L -1}(x)$ does not depend on the $\lg L-1$ least significant
  bits of $x$, hence we can determine the intervals $\ell_{\lg
    L-1}(x)$ and $r_{\lg L -1}(x)$ from $T_j$. We now run the query
  algorithm for $\ell_{\lg L -1}(x)$ and $r_{\lg L-1}(x)$ with $M^*$
  as the memory. If only one of them returns \emph{non-empty}, it
  follows from Lemma~\ref{lem:queryprops} that we have recovered the
  $(\lg L-1)$'st least significant bit of $x$. If both return
  \emph{non-empty}, we read off one bit from the part of the encoding
  written during step 3 of the encoding algorithm. This bit tells us
  whether $x \in \ell_{\lg L-1}(x)$ or $x \in r_{\lg L-1}(x)$ and we have
  again recovered the next bit of $x$. Note that we have reduced the
  number of unknown bits in $x$ by one and we now determine $\ell_{\lg L
    -2}(x)$ and $r_{\lg L-2}(x)$ and recurse. This process continues until all
  bits of $x$ have been recovered and we continue to the next $T_{j'}
  \in T_S$.
\end{enumerate}

Thus we have shown how to encode and decode any $L$-well-separated
point set $S$. We are ready to derive the lower bound.

\paragraph{Analysis.}
The size of the encoding is
\begin{eqnarray*}
s + \lg \binom{A(M^*)+n}{n} + B(M^*) &\leq&\\
s + \lg \binom{\frac{1}{1-1/\gamma} \cdot \eps U / L + n}{n} + \gamma \eps n \lg L &\leq& \\
s + n \lg \left(\frac{e \left(\frac{1}{1-1/\gamma} \cdot \eps U / L + n\right)}{n}\right) + \gamma \eps n \lg L &\leq&\\
s + n \lg \left( \frac{\eps U}{(1-1/\gamma)nL}\right)+ \gamma \eps n \lg L + O(n) &=&\\
s + n \lg \left( \frac{\eps U}{(1-1/\gamma)nL^{1-\eps \gamma}}\right)+O(n)
\end{eqnarray*}
bits. But from Lemma~\ref{lem:manywell} we have that there are at least
$((U-4nL)/n)^n$ distinct $L$-well-separated point sets. Hence we must
have
\begin{eqnarray*}
s + n \lg \left( \frac{\eps U}{(1-1/\gamma)nL^{1-\eps \gamma}}\right)+O(n) &\geq& n \lg \left(\frac{U-4nL}{n}\right) \Rightarrow \\
s &\geq& n \lg \left(\frac{(1-1/\gamma)L^{1-\eps \gamma}(U-4nL)}{\eps U}\right) - O(n).
\end{eqnarray*}
For $L \leq U/5n$, this is:
\begin{eqnarray*}
s &\geq& n \lg \left(\frac{(1-1/\gamma)L^{1-\eps \gamma}}{\eps}\right) - O(n).
\end{eqnarray*}
Setting $\gamma$ to a constant, this simplifies to
$$
s \geq n \lg \left(\frac{L^{1-O(\eps)}}{\eps}\right)-O(n),
$$
which completes the proof of Theorem~\ref{thm:lower}.

\subsection{Extension to two-sided errors}
\label{sec:twosided}
In this section, we extend the lower bound above to the case of two
sided error, i.e. we have both false positives and false
negatives. The result of the section is the following:
\begin{theorem}
For the approximate range emptiness problem on $n$ points from a
universe $[U]$, any data structure answering all query intervals of a
fixed length $L \leq u/5n$ with two-sided error rate $\eps$, must use
at least
$$s \geq n \lg(L/\eps) - O(n)$$
bits of space when $0 < \eps < 1/\lg U$, and at least
$$
s = \Omega\left(\frac{n \lg(L \lg U)}{\lg_{1/\eps} \lg U}\right)
$$
bits for $1/\lg U \leq \eps \leq 1/2 - \Omega(1)$.
\end{theorem}

As in the previous section, we use an approximate range emptiness data
structure to encode any set of $L$-well-separated points. The encoding
procedure follows that of the one-sided error case closely and we
assume the reader has read Section~\ref{sec:lower}.

So assume we are given an approximate range emptiness data structure
$D$ with $s$ bits of space and two-sided error rate $0 < \eps \leq
1/2-\Omega(1)$. Our first step is to reduce the error rate using a
standard trick: Upon receiving an input set of points $S$, implement
$k$ copies of $D$ on $S$, where the randomness used for each copy is
independent. When answering a query, we ask the query on each copy of
the data structure and return the majority answer. Since $\eps <
1/2-\Omega(1)$, it follows from a Chernoff bound that the error rate
is bounded by $\eps^{O(k)}$. Thus we now have a data structure using
$k s$ bits of space with error rate $\delta=\eps^{O(k)}$ for a parameter $k \geq 1$ to be determined later.

\paragraph{Encoding Algorithm.}
Let $S \in \Ins$ be an $L$-well-separated point set to encode. Define
$\Top$, $\ell_i(x)$, $r_i(x)$, $\M$, $A(\M)$ and $B(\M)$ as in
Section~\ref{sec:lower}. Define $C(\M)$ as the number of false
negatives amongst queries in $\Top$ and $F(\M)$ as the number of false
negatives amongst queries $\ell_i(x)$ and $r_i(x)$ over all $i$ and
$x$. From Markov's inequality and a union bound, we conclude that
there must exist a set of memory bits $M^* \in \{0,1\}^{ks}$ for
which:
\begin{itemize}
\item $A(M^*) \leq 4 \delta U/L$.
\item $B(M^*) \leq 4 \delta n \lg L$.
\item $C(M^*) \leq 4 \delta n$.
\item $F(M^*) \leq 4 \delta n \lg L$.
\end{itemize}
The encoding algorithm first writes down some bits that allow us to
correct all the false negatives. Following that, it simply writes down
the encoding from the previous section (since we have reduced the
decoding problem to the case of no false negatives). Describing the
false negatives is done as follows:
\begin{enumerate}
 \item First we write down $\lg(U/L)$ bits to specify
   $C(M^*)$. Following that, we encode the $C(M^*)$ false negatives
   amongst $\Top$ using $\lg \binom{U/L}{C(M^*)}$ bits.
 \item Secondly, we spend $\lg(n \lg L)$ bits to specify $F(M^*)$. We
   then use $\lg \binom{2n\lg L}{F(M^*)}$ bits to specify the false
   negatives amongst $\ell_i(x)$ and $r_i(x)$. Note that we avoid the
   explicit encoding of $\ell_i(x)$ and $r_i(x)$ by mapping a false
   negative $\ell_i(x)$ to the number $\Gamma(\ell_i(x)) = \rank_S(x)\cdot 2 \lg L +
   2(i-1)$ and $r_i(x)$ to the number $\Gamma(r_i(x))=\rank_S(x) \cdot 2 \lg L +
   2(i-1)+1$. Here $\rank_S(x)$ denotes the number of elements in $S$
   smaller than $x$. Note that $\Gamma(\ell_i(x)),\Gamma(r_i(x)) \in [2n\lg L]$ and
   the claimed space bound follows.
\item Lastly, we run the entire encoding algorithm from
  Section~\ref{sec:lower} assuming the false negatives have been
  corrected.
\end{enumerate}

\paragraph{Decoding Algorithm.}
In the following, we describe the decoding algorithm.
\begin{enumerate}
  \item From the bits written in step 1 of the encoding procedure, we
    correct all false negatives amongst queries in $\Top$. 
  \item From the secondary encoding, we recover $M^*$. Since we have
    corrected all false negatives in $\Top$, we also recover $T_S$
    using steps 2-3 of the decoding procedure in
    Section~\ref{sec:lower}.
  \item We now run step 4 of the decoding algorithm in
    Section~\ref{sec:lower}. Observe that for each $x$ we are about to
    recover, we know $\rank_S(x)$. Therefore we can use the bits
    written in step 2 of the encoding algorithm above to correct all
    false negatives amongst queries $\ell_i(x)$ and
    $r_i(x)$. Following that, we can finish step 4 of the decoding
    algorithm in Section~\ref{sec:lower} as all false negatives have
    been corrected.
\end{enumerate}

\paragraph{Analysis.}
Examining the encoding algorithm above and the one in
Section~\ref{sec:lower}, we see that the produced encoding uses:
$$
ks + \lg \binom{U/L}{C(M^*)} + \lg \binom{2n\lg L}{F(M^*)} + \lg \binom{A(M^*)+n}{n} + B(M^*) + O(\lg n)
$$
bits. By the arguments in Section~\ref{sec:lower}, this is bounded by:
$$
ks+\lg \binom{U/L}{C(M^*)} + \lg \binom{2n\lg L}{F(M^*)} + n \lg\left(\frac{\delta U}{4nL^{1-4\delta}}\right)+O(n).
$$
Using the bounds on $C(M^*)$ and $F(M^*)$, we see that this is at most:
$$
ks + 4\delta n \lg \left(\frac{U}{L \delta n}\right) + 4 \delta n \lg L \lg\left(\frac{e}{2 \delta}\right)+ n \lg\left(\frac{\delta U}{4nL^{1-4\delta}}\right)+O(n).
$$
Fixing $k = \max\{1, \Theta(\lg_{1/\eps} \lg U)\}$, the second and third term becomes $O(n)$ and we get:
$$
ks + n \lg \left(\frac{\delta U}{4nL^{1-4\delta}}\right) + O(n).
$$
For $L \leq U/5n$ we have at least $(U/5n)^n$ distinct $L$-well-separated point sets, thus we must have:
\begin{eqnarray*}
ks &\geq& n \lg \left( \frac{L^{1-O(\delta)}}{\delta}\right)-O(n) \Rightarrow \\
s &\geq& \frac{n \lg\left(\frac{L^{1-O(1/\lg U)}}{\min \{ \eps, 1/\lg U\}}\right)}{k}-O(n)\Rightarrow\\
s &\geq& \frac{n \lg\left(\frac{L}{\min \{ \eps, 1/\lg U\}}\right)}{k}-O(n).\\
\end{eqnarray*}
When $\eps < 1/\lg U$ we get $k=1$ and this simplifies to
$$
s = n \lg(L/\eps)-O(n),
$$
and for $\eps > 1/\lg U$ it becomes:
$$
s = \Omega\left(\frac{n \lg(L \lg U)}{\lg_{1/\eps} \lg U}\right).
$$

\section{Upper Bounds}
\label{sec:upper}
In this section we describe our approximate range emptiness data
structures. 
%
The data structure consists of a (non-approximate) succinct range
emptiness data structure applied to the input points after they have
been mapped to a smaller universe. More specifically, we use a
carefully chosen hash function to map the input points to a universe
of size $r=nL/\eps$, and store them in a range emptiness data
structure that has constant query time and uses
$n\log(U/n)+O(n\log^\delta (U/n))$ bits of space for storing $n$
points from a universe of size $U$.
This gives us our main upper bound result.
\begin{theorem}
\label{thm:lower}
For any $\varepsilon>0,L$ and $n$ input points from a universe $[U]$,
there is a data structure that uses $n\log(L/\varepsilon) +
O(n\log^\delta (L/\varepsilon))$ bits of space, where $\delta>0$ is
any desired constant, that answers range emptiness queries for all
ranges of length at most $L$ in constant time with a false positive rate of at most
$\eps$.
\end{theorem}

\subsection{Universe Reduction}
Let $r = n L / \varepsilon$ and choose $u: [U/r] \rightarrow [r]$ from a pairwise independent family.
Now define a hash function (similar in spirit to the ``non-expansive'' hash functions of Linial and Sasson~\cite{linial}) that preserves locality, yet has small collision probability:
\[ h(x) = (u(\lfloor x/r \rfloor) + x) \mod r \enspace . \]
\begin{lemma}
For $x_1\ne x_2$ we have $\Pr[h(x_1)=h(x_2)] \leq 1/r$. 
\end{lemma}
\begin{proof}
Collision happens if and only if $u(\lfloor x_1 / r\rfloor) - u(\lfloor x_2 / r\rfloor) \equiv x_2 - x_1  \mod r $. 
If $\lfloor x_1 / r\rfloor = \lfloor x_2 / r \rfloor $ then $x_2 - x_1 \not\equiv 0 \mod r$, so the collision probability is zero.
Otherwise, since $u$ is pairwise independent $u(\lfloor x_1 / r \rfloor) - u(\lfloor x_2 / r \rfloor)$ is equal to a given fixed value with probability exactly $1/r$.
\end{proof}

We store $h(S)\subseteq [r]$ in a (non-approximate) range emptiness data structure.
To answer an approximate range membership query on an interval $I$, observe that the image $h(I)$ 
will be the union of at most two intervals $I_1,I_2\subseteq [r]$.
If either of these intervals are non-empty in $h(S)$ we report ``not empty'', otherwise we report
``empty''.
It is clear that there can be no false negatives.
False positives occur when $I\cap S = \emptyset$ but there is a hash
collision among a point $x\in S$ and a point $y\in I$.
We can bound the false positive rate by a union bound over all possible collisions:
\[ \sum_{x\in S} \sum_{y\in I} \Pr[h(x)=h(y)] \leq n L / r \leq \varepsilon \enspace . \]

\subsection{Range Emptiness Data Structure}

We first describe a range emptiness data structure for size-$n$
subsets of $[U]$ that answers queries in constant time and uses $n \lg
U + O(n \lg^\delta U)$ bits, where $\delta > 0$ is any desired
constant.  Later, we show how to decrease the space usage to $n
\lg(U/n) + O(n \lg^\delta (U/n))$ bits.  The data structure consists
of the sorted list of points plus an indexing data structure,
namely, the weak prefix search data structure of~\cite[Theorem 5,
  second part]{weakprefix}. A weak prefix query on a set of points in $[U]$
is specified by a bit string $p$ of length at most $\log U$ and returns the
interval of ranks of the input points that have prefix $p$ (when
written in binary). If no such points exist the answer is arbitrary.
Given a query range $[a ; b]$ we compute the longest common prefix $p$
of the bit representations of $a$ and $b$.  This is possible in $O(1)$
time using a most significant bit computation.  Observe as
in~\cite{alstrup1d} that $h(S)\cap [a;b]$ is
non-empty if and only if at least one of the following holds:
\begin{itemize}
	\item A largest point in $h(S)$ prefixed by $p\circ 0$ exists, and is not smaller than $a$, or
	\item A smallest point in $h(S)$ prefixed by $p\circ 1$ exists, and is not greater than $b$.
\end{itemize}
We can determine if each of these holds by a weak prefix query, by
considering the points in the sorted list at the position of the
maximum and minimum returned ranks.  If there are no points with
prefix $p\circ 0$ or $p\circ 1$ the range returned by the weak prefix
search is arbitrary, but this is no problem since we can always check
points for inclusion in $[a;b]$.  The space usage for the weak
prefix search data structure, in the case of constant query time, can
be made $O(n \lg^\delta (U))$ bits for any constant $\delta>0$.

In order to reduce the space usage, we use a standard trick and split
the universe $[U]$ into $n$ subranges $s_1,\ldots,s_n$ of size $U/n$.
We need the well known succinct rank and select data structures that
stores a bit array of size $n$ using $n+o(n)$ bits of space and
supports rank and select queries in constant time \cite{succincter}.
We construct a bit array of size $n$ that has a one at position $i$ if
there is an input point in the range $s_i$ and zero otherwise, and
store it in a rank/select data structure $D_1$.
We store the data structure from above for each non-empty range $s_i$ (with
universe size $U/n$) using exactly $n_i (\log(U/n) +
\alpha(\log^\delta (U/n)))$ space where $n_i$ is the number of points in
$s_i$ and $\alpha>0$ is a constant that depends on the data structure.
The data structures are stored consecutively in an array $A_{ds}$.  To
locate the data structure for any range $s_i$, all we need is to count
the number of points in the ranges $s_j$ for $j<i$ and scale that
number accordingly.
For this purpose we store another array of $2n$ bits that for each
non-empty range $s_i$ stores a one followed by $n_i$ zeros in a
rank/select data structure $D_2$. For a given range $s_i$ we can
compute the starting position for the corresponding data structure in
$A_{ds}$ as follows. We compute the number of non-empty ranges $s_j$
with $j<i$ using a rank query for $i$ in $D_1$. Then we do a select
query in $D_2$ for the returned rank and subtract the queried rank
from the result to get the number of points in $s_j$ with
$j<i$. Finally we scale this number with $(\log(U/n) +
\alpha(\log^\delta (U/n)))$.
The total space usage becomes $n \log (U/n) + O(n \log^{\delta}
(U/n))$ bits.

A query range is naturally split into at most three parts, a part
consisting of a sequence of ranges, and maximally two parts that do
not cover an entire range. If we find a point in any of them we
report ``not empty'' and ``empty'' otherwise.  The part spanning
entire ranges is answered by computing the rank difference (number of
input points) between the endpoint in the rightmost range and the
starting index in the leftmost range spanned by the query with the
rank/select data structure. 
The non-spanning parts of the query are answered by the data structures
stored for the corresponding subranges (which we locate in $A_s$ as
described above). The query takes constant time, since answering each
part takes constant time.

\paragraph{Extension to Range Reporting.}
We note that the above data structure supports reporting all $k$
points inside a query interval in $O(k)$ time: Observe first that when
the above data structure returns ``not empty'', it actually finds a
point $p$ inside the query interval as well as $p$'s location in the sorted list
of all points in $p$'s subrange. We scan these points in
both the left and right direction, starting at $p$, and stop when a
point outside the query range is encountered. If all points preceeding
(or following) $p$ in the subrange are reported, we can find the next
subrange to report from using the rank and select data structure $D_1$ on
the non-empty subranges (we know the rank of the current subrange). We
conclude that we spend $O(1)$ time per reported point and thus $O(k)$
time in total.

\bibliographystyle{abbrv}
\bibliography{biblio.bib}

\end{document}